\documentclass[11pt]{article}

\usepackage{amsthm, amsmath, amssymb, amsbsy, amscd, amsfonts, latexsym, euscript,epsf}
\usepackage{authblk}
\newtheorem{thm}{Theorem}[section]
\usepackage{url}
\allowdisplaybreaks
\newtheorem{proposition}[thm]{Proposition}
\newtheorem{corollary}[thm]{Corollary}

\newtheorem{theorem}[thm]{Theorem}

\theoremstyle{definition}

\newtheorem{example}[thm]{Example}
\newtheorem{remark}[thm]{Remark}

\newcommand{\pyb}{\mathrm{YB}_\Omega}
\newcommand{\yb}{\mathrm{YB}}

\DeclareMathOperator{\rank}{rank}
\DeclareMathOperator{\End}{End}

\DeclareMathOperator{\mat}{Mat}
\DeclareMathOperator{\GL}{GL}
\DeclareMathOperator{\id}{id}

\newcommand{\ts}{\mathrm{T}}
\newcommand{\qqquad}{\qquad\quad}

\newcommand{\ybm}{\mathbf{Y}}
\newcommand{\fmp}{\mathbf{F}}

\newcommand{\dft}{\mathrm{d}}

\newcommand{\md}{\mathbf{G}}
\newcommand{\mm}{\mathbf{M}}

\newcommand{\zsp}{\mathbb{Z}_{>0}}

\newcommand{\lb}{\label}
\newcommand{\er}{\eqref}
\newcommand{\cl}{\colon}
\newcommand{\beq}{\begin{gather}}
\newcommand{\ee}{\end{gather}}

\title{\textbf{Yang--Baxter maps, Darboux transformations, and linear approximations of refactorisation problems}}

\author[2]{V.M. Buchstaber\thanks{buchstab@mi-ras.ru}}
\author{S. Igonin\thanks{s-igonin@yandex.ru}}
\author{S. Konstantinou-Rizos\thanks{skonstantin84@gmail.com}}
\author{M.M. Preobrazhenskaia\thanks{rita.preo@gmail.com}}
\affil{Centre of Integrable Systems, P.G. Demidov Yaroslavl State 
University, Russia}
\affil[2]{Steklov Mathematical Institute of Russian Academy of Sciences, Moscow, Russia}

\begin{document}

\maketitle

\begin{abstract}
Yang--Baxter maps (YB maps) are set-theoretical solutions to the quantum Yang--Baxter equation.
For a set $X=\Omega\times V$, where $V$ is a vector space and 
$\Omega$ is regarded as a space of parameters, a linear parametric YB map is a YB map
$Y\cl X\times X\to X\times X$ such that $Y$ is linear with respect to $V$
and one has $\pi Y=\pi$ for the projection $\pi\cl X\times X\to\Omega\times\Omega$.
These conditions are equivalent to certain nonlinear algebraic relations 
for the components of $Y$.
Such a map $Y$ may be nonlinear with respect to parameters from $\Omega$.

We present general results on such maps, 
including clarification of the structure of the algebraic relations that define them
and several transformations which allow one to obtain new such maps from known ones.
Also, methods for constructing such maps are described. 
In particular, developing an idea 
from [Konstantinou-Rizos S and Mikhailov A V 2013 J. Phys. A: Math. Theor. 46 425201],
we demonstrate how to obtain linear parametric YB maps 
from nonlinear Darboux transformations of some Lax operators 
using linear approximations of matrix refactorisation problems corresponding to Darboux matrices.
New linear parametric YB maps with nonlinear dependence on parameters are presented.
\end{abstract}

\maketitle

\bigskip




\section{Introduction}
\label{sintr}

The quantum Yang--Baxter equation is one of 
the most fundamental equations in mathematical physics 
and has applications in many diverse branches of physics and mathematics, 
including statistical mechanics, quantum field theories, the theory of knots and braids.
Set-theoretical solutions of the quantum Yang-Baxter (YB) equation have been intensively studied
by many authors after the work of Drinfeld~\cite{Drin92}.
Even before that, examples of such solutions were constructed by Sklyanin~\cite{skl88}.
Veselov~\cite{Veselov} proposed the term ``Yang--Baxter maps'' (YB maps)
for such solutions and initiated the study of the dynamical aspects for them.
Before~\cite{Veselov}, one of us~\cite{Buchstaber} introduced the term ``Yang--Baxter mappings''
in the study of algebraic aspects for such maps.
Relations of YB maps with integrable systems (including integrable PDEs and lattice equations) is a very active area of research 
(see, e.g.,~\cite{abs2003,ABS-2005,Vincent,dim-mull,hjn-book,Pavlos2019,Pavlos,Sokor-Sasha,
Kouloukas, Kouloukas2,PT,PTV,PSTV,Veselov2,Veselov3} and references therein).

Let $V$ be a set.
A YB map $Y:V\times V\to V\times V$ is 
a set-theoretical solution to the quantum YB equation
\begin{gather}
\label{ybeq}
Y^{12}\circ Y^{13}\circ Y^{23}=Y^{23}\circ Y^{13}\circ Y^{12}.
\end{gather}
Here $Y^{ij}\cl V\times V\times V\to V\times V\times V$, $\,i,j=1,2,3$, $i<j$, 
is the map acting as $Y$ on the $i$th and $j$th factors of $V\times V\times V$ 
and acting as identity on the remaining factor.

A parametric YB map 
\begin{gather}
\label{ParamYB}
Y_{a,b}\colon V\times V\to V\times V,\qquad
Y_{a,b}(x,y)=\big(u_{a,b}(x,y),\,v_{a,b}(x,y)\big),\qquad x,y\in V,
\qquad a,b\in\Omega,
\end{gather}
depends on parameters $a,b\in\Omega$ from some parameter set $\Omega$ and obeys 
the parametric YB equation
\begin{gather}
\label{YB_eq}
Y^{12}_{a,b}\circ Y^{13}_{a,c} \circ Y^{23}_{b,c}=
Y^{23}_{b,c}\circ Y^{13}_{a,c} \circ Y^{12}_{a,b}\qquad\text{for all }\,a,b,c\in\Omega.
\end{gather}
The precise definition of such maps is given in Section~\ref{sybm}.

In this paper we mostly study the case when $V$ is a vector space over a field~$\mathbb{K}$ and 
for all $a,b\in\Omega$ the map $Y_{a,b}$ given by~\er{ParamYB} is $\mathbb{K}$-linear. 
Note that usually $\Omega$ is an open subset of another vector space, 
and the dependence of~$Y_{a,b}$ on the parameters~$a,b$ is nonlinear.
(See also Remark~\ref{rrdp} on the case when $\Omega$ is an algebraic variety.)

Examples of such maps related to integrable PDEs of 
Kadomtsev--Petviashvili (KP) and nonlinear Schr\"odinger (NLS) types
can be found in~\cite{dim-mull,Sokor-Sasha}. These examples 
are discussed in Example~\ref{exdmh} and in Section~\ref{saddnls} of the present paper.
\begin{remark}
\lb{rnonl}
A parametric YB map~\er{ParamYB} with parameters $a,b\in\Omega$
can be interpreted as the following YB map~$\ybm$ without parameters
\begin{gather}
\lb{ypsv}
\begin{gathered}
\ybm\cl (\Omega\times V)\times (\Omega\times V)\to(\Omega\times V)\times (\Omega\times V),\\
\ybm\big((a,x),(b,y)\big)=\big((a,u_{a,b}(x,y)),(b,v_{a,b}(x,y))\big),
\end{gathered}
\end{gather}
which satisfies $\pi\ybm=\pi$ for the projection $\pi\cl (\Omega\times V)\times (\Omega\times V)\to\Omega\times\Omega$.
Thus one can say that 
we study YB maps of the form~\er{ypsv} which are linear with respect to~$V$ 
and may be nonlinear with respect to~$\Omega$.
However, often it is useful to keep $a,b$ as parameters and to work with~$Y_{a,b}$.
\end{remark}

Note that we do not impose any nondegeneracy conditions on YB maps.
In particular, the maps are not required to be invertible.

For a vector space $V$, a parametric map~\er{ParamYB} is linear if and only if it is of the form
\begin{gather}
\label{ilpyb}
Y_{a,b}\colon V\times V\to V\times V,\qquad
Y_{a,b}(x,y)=\big(A_{a,b}x+B_{a,b}y,\,C_{a,b}x+D_{a,b}y\big),\qquad x,y\in V,
\qquad a,b\in\Omega,
\end{gather}
for some linear maps $A_{a,b},B_{a,b},C_{a,b},D_{a,b}$ from $V$ to $V$, 
which may depend nonlinearly on $a,b\in\Omega$.

In Proposition~\ref{ABCD} we prove that a linear parametric map~\er{ilpyb}
satisfies the parametric YB equation~\eqref{YB_eq} (i.e., \er{ilpyb}
is a parametric YB map) if and only if $A_{a,b},B_{a,b},C_{a,b},D_{a,b}$
obey the algebraic relations \eqref{ABC-relations-a}--\eqref{ABC-relations-e} 
for all values of the parameters $a,b,c\in\Omega$.
The case when $A_{a,b},B_{a,b},C_{a,b},D_{a,b}$ do not depend on $a,b$ 
was studied in~\cite{Buchstaber,ess99} (see Remark~\ref{rnoab} for more details).

Let $\pyb(V)$ be the set of linear parametric YB maps~\er{ilpyb}. 
An element $Y_{a,b}\in\pyb(V)$ given by~\eqref{ilpyb} 
is written as $Y_{a,b}=(A_{a,b},D_{a,b};B_{a,b},C_{a,b})\in\pyb(V)$.
In Theorem~\ref{thptr} we describe several transformations which
allow one to obtain new such maps from a given one.

From Theorem~\ref{thptr} we deduce Corollary~\ref{clc}, namely, 
if $(A_{a,b},D_{a,b};B_{a,b},C_{a,b})\in\pyb(V)$, 
then for any nonzero constant $l\in\mathbb{K}$ one has 
$(lA_{a,b},l^{-1}D_{a,b};B_{a,b},C_{a,b})\in\pyb(V)$.
Thus, from a given linear parametric YB map~\er{ilpyb},
we derive a family of linear parametric YB maps~\er{lald} depending on~$l$.

Theorem~\ref{thme} contains a number of identities for linear parametric YB maps.
As discussed in Remark~\ref{rpequiv}, this clarifies the structure of relations
\eqref{ABC-relations-a}--\eqref{ABC-relations-e}.

Applying Corollary~\ref{clc} to the YB map~\er{mhm} from~\cite{dim-mull},
for each nonzero $l\in\mathbb{C}$ we obtain a linear parametric YB map~\er{ylab},
which for $l\neq 1$ is new.

It is known that one can sometimes construct rational parametric YB maps 
from matrix refactorisation problems corresponding to Darboux matrices for 
Lax operators (see, e.g.,~\cite{Sokor-Sasha}).
However, usually such a refactorisation problem gives an algebraic variety~\cite{Sokor-Sasha}, 
and it is often difficult to represent the variety as the graph of a rational map; 
namely, one cannot always solve the corresponding polynomial equations in order to define a map.
Then one can try to find a linear approximation to the latter.
Developing an idea from~\cite{Sokor-Sasha} (see Remark~\ref{rmikh} below),
in Section~\ref{NLS-type-maps} we study linear approximations of such
refactorisation problems and present examples where this gives linear parametric YB maps.
We present several explicit examples of this procedure for Lax operators of NLS type.
It would be interesting to apply the described method to Lax operators of other types.

\begin{remark}
\lb{rmikh}
The idea to use linear approximations of Darboux matrix refactorisation problems 
for obtaining linear parametric YB maps appeared in~\cite{Sokor-Sasha},
where one example of such a map was obtained (the map~\er{ybmss} in the present paper).

Developing this idea, in Section~\ref{NLS-type-maps} 
we present a detailed description of a method to obtain such maps.
In particular, in Remarks~\ref{reps1},~\ref{reps2},~\ref{reps3}, using Proposition~\ref{paeps},
we explain why the linear parametric maps derived from the considered
linear approximations of Darboux matrix refactorisation problems
satisfy the parametric YB equation. 
In~\cite{Sokor-Sasha} this is not explained.

As said above, the map~\er{ybmss} was obtained in~\cite{Sokor-Sasha}.
Applying Corollary~\ref{clc} to the map~\eqref{ybmss}, 
for each nonzero constant $l\in\mathbb{C}$ we derive the linear parametric YB map~\er{lybmss}, 
which for $l\neq 1$ is new.
\end{remark}

In Section~\ref{sybmr}, after recalling the well-known relation between matrix refactorisation problems 
and YB maps (see, e.g., \cite{Veselov,Veselov2,Kouloukas}),
in Proposition~\ref{pasa} we present a straightforward generalisation of this relation
to the case of refactorisation problems in arbitrary associative algebras.
Using Proposition~\ref{pasa}, we prove Proposition~\ref{paeps}, 
which explains why the linear parametric maps 
arising from the linear approximations of the matrix refactorisation problems
considered in Section~\ref{NLS-type-maps} obey the parametric YB equation~\eqref{YB_eq}
(see Remarks~\ref{reps1},~\ref{reps2},~\ref{reps3}).

As shown in the proof of Proposition~\ref{paeps}, 
for a given associative algebra $\mathfrak{B}$, this proposition 
deals with refactorisation in the associative algebra
of formal sums $m_1+\varepsilon m_2$, where $m_1,m_2\in\mathfrak{B}$ 
and $\varepsilon$ is a formal symbol assumed to satisfy~$\varepsilon^2=0$.

In Section~\ref{sYBmvb} we present the following example of constructing linear parametric YB maps (with nonlinear dependence on parameters) from a nonlinear nonparametric YB map.

For any group $\md$, one has the nonlinear nonparametric YB map~\er{fxyx}
from~\cite{Drin92}.
Taking $\md=\GL_n(\mathbb{K})\subset\mat_n(\mathbb{K})$ for a positive integer~$n$, 
for any abelian subgroup $\Omega\subset\md$ and 
any nonzero constant $l\in\mathbb{K}$ we construct from~\er{fxyx}
a new linear parametric YB map~\er{lpyblg}.
The construction consists of three steps:
\begin{itemize}
\item First assume that $\mathbb{K}$ is either $\mathbb{R}$ or $\mathbb{C}$.
Then the set $\md\times\md=\GL_n(\mathbb{K})\times\GL_n(\mathbb{K})$
is a manifold, the YB map~$\fmp$ given by~\er{fxyx} is an analytic diffeomorphism, 
hence we can consider the differential~$\dft\fmp$ of~$\fmp$, which is given by~\er{dfxyx},
where $\mm=\mat_n(\mathbb{K})$.
Since $\fmp$ is a YB map, its differential $\dft\fmp$ is a YB map as well.
\item 
Let $\Omega\subset\md$ be an abelian subgroup of~$\md$.
Denote by $\ybm$ the restriction of the map~\er{dfxyx} to the subset
$(\Omega\times\mm)\times(\Omega\times\mm)\subset(\md\times\mm)\times(\md\times\mm)$.
Computing~$\ybm$, we derive the YB map~\er{ybmgl}, which can be interpreted 
as the linear parametric YB map~\er{pyblg} with parameters $a,b\in\Omega$.
To obtain~\er{ybmgl} from~\er{dfxyx}, 
we use the fact that $\fmp(a,b)=(a,aba^{-1})=(a,b)$ for any $a,b\in\Omega\subset\md$, 
because $ab=ba$.
\item For each nonzero $l\in\mathbb{K}$, 
applying Corollary~\ref{clc} to the map~\eqref{pyblg}, 
we get the linear parametric YB map~\er{lpyblg} 
with nonlinear dependence on the parameters $a,b$.
In the described construction of~\er{lpyblg} we have assumed
that $\mathbb{K}$ is either $\mathbb{R}$ or $\mathbb{C}$, 
but one can check that \er{lpyblg} is a parametric YB map for any field~$\mathbb{K}$.
\end{itemize}
Thus, \er{lpyblg} represents an infinite collection of new 
linear parametric YB maps corresponding to 
abelian subgroups~$\Omega$ of the matrix group $\GL_n(\mathbb{K})$.
A detailed description of this construction is given in Section~\ref{sYBmvb}.

In Section~\ref{sgener}, in the framework of fibre bundles and vector bundles, 
we discuss certain generalisations
of some notions and constructions considered in this paper.

Section~\ref{sconc} concludes the paper 
with some remarks and comments on how the results of this paper can be extended.

\begin{remark}
\lb{rinvol}
Recall that, for a set $S$, a map $F\cl S\to S$ 
is said to be \emph{involutive} if $F^2=\id$. 

In the study of the dynamics of a map $\tilde{F}\cl S\to S$, 
one considers its powers $\tilde{F}^k\cl S\to S$ for nonnegative integers $k$.
From this point of view, noninvolutive maps are more interesting than involutive ones.

Known examples of YB maps very often turn out to be involutive.
The maps~\er{ylab},~\er{lybmss},~\er{lpyblg} are noninvolutive.
\end{remark}

\begin{remark}
\lb{rnoab}
As a preparation for the study of linear parametric YB maps,
in Section~\ref{swp} we recall some results (mostly from~\cite{Buchstaber})
on linear YB maps without parameters.
In particular, the result of Proposition~\ref{pabcd} on the structure of
linear YB maps without parameters appeared indepedently in~\cite{Buchstaber} and~\cite{ess99}
in different notation described in Remark~\ref{rchnot}.
As noticed in~\cite{ess99}, the result of Proposition~\ref{pabcd} can also be easily deduced from formulas in~\cite{hi97}.

As discussed in Remark~\ref{rreduce} and in Section~\ref{swp},
in the case when $A_{a,b},B_{a,b},C_{a,b},D_{a,b}$ do not depend on~$a,b$
some results of Section~\ref{slpYBm} reduce to ones from~\cite{Buchstaber}.
\end{remark}

\section{General results on linear parametric Yang--Baxter maps}
\lb{sgr}

\subsection{Yang--Baxter maps}
\label{sybm}

Let $V$ be a set. 
A \textit{Yang--Baxter \textup{(}YB\textup{)} map} is a map 
$$
Y\colon V\times V\to V\times V,\qquad Y(x,y)=(u(x,y),v(x,y)),\qquad x,y\in V,
$$
satisfying the YB equation~\eqref{ybeq}.
The terms $Y^{12}$, $Y^{13}$, $Y^{23}$ in~\eqref{ybeq} are maps $V\times V\times V\to V\times V\times V$ defined as follows 
\begin{gather}
\lb{y1223}
Y^{12}(x,y,z)=\big(u(x,y),v(x,y),z\big),\qqquad
Y^{23}(x,y,z)=\big(x,u(y,z),v(y,z)\big),\\
\lb{y13}
Y^{13}(x,y,z)=\big(u(x,z),y,v(x,z)\big),\qqquad x,y,z\in V.
\end{gather}

A \textit{parametric YB map} $Y_{a,b}$ is a family of maps~\er{ParamYB}
depending on two parameters $a,b\in\Omega$ from some parameter set~$\Omega$ 
and satisfying the parametric YB equation~\eqref{YB_eq}.
The terms $Y^{12}_{a,b}$, $Y^{13}_{a,c}$, $Y^{23}_{b,c}$
in~\eqref{YB_eq} are maps $V\times V\times V\to V\times V\times V$ 
defined similarly to~\er{y1223},~\er{y13}, adding the parameters $a,b,c$.
For instance, 
$$
Y^{13}_{a,c}(x,y,z)=\big(u_{a,c}(x,z),y,v_{a,c}(x,z)\big),\qqquad
x,y,z\in V.
$$
In general, $V$ and $\Omega$ can be arbitrary sets.
See also Remark~\ref{rrdp} on the case when $\Omega$ is an algebraic variety.

\begin{remark}
\lb{rrdp}
Suppose that $\Omega$ is an algebraic variety and consider the Zariski topology
on the algebraic variety $\Omega\times\Omega$.
Quite often one meets the following situation.
A parametric map $Y_{a,b}$ depends rationally on parameters $a,b$ from~$\Omega$, 
and there is an open dense subset $W\subset\Omega\times\Omega$ such 
that $Y_{a,b}$ is defined for all $(a,b)\in W\subset\Omega\times\Omega$.

If such $Y_{a,b}$ satisfies the parametric YB equation
$Y^{12}_{a,b}\circ Y^{13}_{a,c} \circ Y^{23}_{b,c}=
Y^{23}_{b,c}\circ Y^{13}_{a,c} \circ Y^{12}_{a,b}$
for all points $(a,b)$, $(a,c)$, $(b,c)$ from $W\subset\Omega\times\Omega$,
then $Y_{a,b}$ can be called a parametric YB map. 
(Although $Y_{a,b}$ may be undefined for some points $(a,b)\in(\Omega\times\Omega)\setminus W$.)
Results and methods of this paper are valid for such maps.

For example, the YB map~\er{ybmss} from~\cite{Sokor-Sasha} is of this type,
because \er{ybmss} is undefined for $(a,b)$ satisfying $a+b=0$.
\end{remark}

\subsection{Notation}
\label{snot}

Let $V$ be a vector space over a field $\mathbb{K}$.
By $\End(V)$ we denote the space of linear maps $V\to V$. 
The space $\End(V)$ is an associative algebra with respect to composition of maps.

Let $A, B, C, D \in\End(V)$. 
Consider the vector space ${V}\times {V}\cong V\oplus V$.
The linear map
\begin{gather}
\notag
Y\colon{V}\times {V}\to{V}\times {V},\qqquad
Y(x,y)=(Ax+By,\,Cx+Dy),\qqquad x,y\in V,
\end{gather}
is written as $Y=(A,D;B,C)\in\End(V\times V)$, where 
$\End(V\times V)$ is the space of linear maps ${V}\times {V}\to{V}\times {V}$.

$\zsp$ is the set of positive integers. Let $n\in\zsp$.
For any commutative algebra $\mathcal{A}$, we denote
by $\mat_n(\mathcal{A})$ the associative algebra of $n\times n$ matrices 
with entries from~$\mathcal{A}$.

\subsection{Linear Yang--Baxter maps without parameters}
\label{swp}

Let $V$ be a vector space over a field $\mathbb{K}$.
In this subsection, we study the YB equation~\eqref{ybeq}
for a linear map $Y\colon{V}\times {V}\to{V}\times {V}$ given by
\begin{gather}
\label{yabcd}
Y(x,y)=(Ax+By,\,Cx+Dy),\qquad x,y\in V,\qquad A,B,C,D\in\End(V).
\end{gather}
So here we consider linear YB maps without parameters.
Many examples of such maps can be found in~\cite{Buchstaber,ess99}.
Results presented in this subsection serve as a preparation for the study 
of linear parametric YB maps in Subsection~\ref{slpYBm}.

Let $\yb(V)\subset\End(V\times V)$ be the subset of maps $Y\in\End(V\times V)$ satisfying~\eqref{ybeq}. 
An element $Y\in\yb(V)$ given by~\eqref{yabcd} is written as $Y=(A,D;B,C)\in\yb(V)$.

The result of Proposition~\ref{pabcd} is presented in~\cite{Buchstaber,ess99} 
in different notation described in Remark~\ref{rchnot} below.
As noticed in~\cite{ess99}, this result can also be easily deduced from formulas in~\cite{hi97}.

\begin{proposition}[\cite{Buchstaber,ess99}]
\label{pabcd}
A map $Y\in\End(V\times V)$ given by \eqref{yabcd} satisfies~\eqref{ybeq}
if and only if the maps $A,B,C,D\in\End(V)$ in~\eqref{yabcd} obey the following relations 
\begin{subequations}
\begin{gather}
\label{lybr1}
C^2=C-DCA, \qqquad B^2=B-ABD, \\
\label{lybr2}
DC-CD=DCB,\qqquad AB-BA=ABC,\\
\label{lybr3}
CA-AC=BCA,\qqquad  BD-DB=CBD,\\
\label{lybr4}
DA-AD=BCB-CBC.
\end{gather}
\end{subequations}
Thus, for $A,B,C,D\in\End(V)$ we have $(A,D;B,C)\in\yb(V)$ 
if and only if $A$, $B$, $C$, $D$ obey relations~\eqref{lybr1}--\eqref{lybr4}.
\end{proposition}

\begin{remark}
\label{rchnot}
In~\cite{Buchstaber,ess99} the following braid relation
\begin{gather}
\label{brr}
\tilde{Y}^{12}\circ \tilde{Y}^{23}\circ\tilde{Y}^{12}=
\tilde{Y}^{23}\circ \tilde{Y}^{12}\circ\tilde{Y}^{23}
\end{gather}
is studied instead of the YB equation~\eqref{ybeq}.

Consider the permutation map $P\in\End(V\times V)$, $P(x,y)=(y,x)$.
A map $Y\in\End(V\times V)$ given by~\eqref{yabcd} satisfies~\eqref{ybeq} if and only if 
the map $\tilde{Y}=PY$ satisfies~\eqref{brr}.

From~\eqref{yabcd} one obtains $PY(x,y)=(Cx+Dy,\,Ax+By)$ for $x,y\in V$.
Because of this, formulas~\eqref{lybr1}--\eqref{lybr4} appear in~\cite{Buchstaber,ess99} 
in different notation: $A$ is interchanged with $C$, and $B$ is interchanged with $D$.

Equations~\er{cbcbd}, \er{dbdc}, \er{bmcdab} below  
and the results of Proposition~\ref{cmeb} appear 
in~\cite{Buchstaber} with the same notation change.
\end{remark}

The following relations are presented in~\cite{Buchstaber} as consequences of 
relations~\er{lybr1}--\er{lybr4}
\begin{gather}
\lb{cbcbd}
[C+B-CB,D]=0,\qqquad [C+B-BC,A]=0,\\
\lb{dbdc}
[D-BDC,A]=[BC,C+B],\qqquad [A-CAB,D]=[CB,C+B],\\
\label{bmcdab}
\begin{pmatrix} 
C & D\\ 
A & B	
\end{pmatrix}
\begin{pmatrix} 
C & BD\\ 
CA & B	
\end{pmatrix}=
\begin{pmatrix} 
C & BD\\ 
CA & B	
\end{pmatrix},
\end{gather}
where $[\cdot,\cdot]$ denotes the commutator of maps in $\End(V)$.
Note that relations~\er{cbcbd} are equivalent to
\begin{gather*}
[(C-\id)(B-\id),D]=0,\qqquad [(B-\id)(C-\id),A]=0,
\end{gather*}
where $\id\cl V\to V$ is the identity map.

The result of Proposition~\ref{thtr} below is presented in~\cite{Buchstaber}, 
using the notation described in Remark~\ref{rchnot}.
This result follows from Proposition~\ref{rchnot}.

\begin{proposition}[\cite{Buchstaber}]
\label{thtr}
For any vector space $V$ over a field $\mathbb{K}$,
the set $\yb(V)$ is invariant under the following transformations
\begin{gather}
\label{dacb}
(A,D;B,C)\mapsto (D,A;C,B),\\
\label{ladl}
\begin{gathered}
(A,D;B,C)\mapsto (LA,DL^{-1};B,C)\\
\text{for any invertible $L\in\End(V)$ that commutes with $B$, $C$, $AD$}.
\end{gathered}
\end{gather}
That is, if $(A,D;B,C)\in\yb(V)$ then $(D,A;C,B)\in\yb(V)$ and $(LA,DL^{-1};B,C)\in\yb(V)$.

Let $V=\mathbb{K}^n$ for some $n\in\zsp$. 
The space $\End(\mathbb{K}^n)$ is identified with $\mat_n(\mathbb{K})$, and we have $\yb(\mathbb{K}^n)\subset\big(\mat_n(\mathbb{K})\big)^4$.
The set $\yb(\mathbb{K}^n)$ is invariant also under the transformation
\begin{gather}
\label{ttabcd}
(A,D;B,C)\mapsto (D^\ts,A^\ts;B^\ts,C^\ts),
\end{gather}
where $\ts$ denotes the transpose operation in $\mat_n(\mathbb{K})$.
\end{proposition}

\begin{remark}
\lb{rctr}
Taking the composition of the transformations~\er{ttabcd},~\er{dacb}, 
we see that the set $\yb(\mathbb{K}^n)$ is invariant also under the transformation
$(A,D;B,C)\mapsto (A^\ts,D^\ts;C^\ts,B^\ts)$.
\end{remark}

For completeness, we present a proof for the following result from~\cite{Buchstaber}, 
which is stated in~\cite{Buchstaber} without proof as a consequence of~\eqref{bmcdab}.
\begin{proposition}[\cite{Buchstaber}]
\label{cmeb}
Let $n\in\zsp$.
Let $A,B,C,D\in\mat_n(\mathbb{K})$ such that $(A,D;B,C)\in\yb(\mathbb{K}^n)$.
Then every nonzero column of the matrix 
$
\begin{pmatrix} 
C & BD \\ 
CA & B	
\end{pmatrix}
$ 
is an eigenvector of the matrix 
$
\begin{pmatrix} 
C & D\\ 
A & B	
\end{pmatrix}
$ 
with eigenvalue~$1$.

Consider the vector space
\begin{gather}
\notag
W=\left\{\begin{pmatrix} x\\ y \end{pmatrix}
\in\mathbb{K}^n\times \mathbb{K}^n\;\left|\;
\begin{pmatrix} 
C & D\\ 
A & B	
\end{pmatrix}
\begin{pmatrix} x\\ y \end{pmatrix}=\begin{pmatrix} x\\ y \end{pmatrix}
\right\}\right..
\end{gather}
We have $\dim W\ge\max(\rank B,\,\rank C)$.

Furthermore, if $\begin{pmatrix} 
C & D\\ 
A & B	
\end{pmatrix}\neq
\begin{pmatrix} 
\mathbf{1}_n & 0\\ 0 & \mathbf{1}_n	
\end{pmatrix}$, 
then $
\det\begin{pmatrix} 
C & BD\\ CA & B	
\end{pmatrix}=0
$.
Here $\mathbf{1}_n$ is the $n\times n$ identity matrix.
\end{proposition}
\begin{proof}
The first statement follows immediately from~\eqref{bmcdab}.

Thus every column of the matrix 
$
\begin{pmatrix} 
C & BD\\ 
CA & B	
\end{pmatrix}
$ 
belongs to $W$. This implies the second statement, 
since $\rank \begin{pmatrix} 
C & BD\\ 
CA & B	
\end{pmatrix} \ge \max(\rank B,\,\rank C)$.

If the matrix $\begin{pmatrix} 
C & BD\\ CA & B	
\end{pmatrix}$ is invertible, then~\eqref{bmcdab} yields 
$\begin{pmatrix} 
C & D\\ 
A & B	
\end{pmatrix}=
\begin{pmatrix} 
\mathbf{1}_n & 0\\ 0 & \mathbf{1}_n	
\end{pmatrix}$.
Therefore, if $\begin{pmatrix} 
C & D\\ 
A & B	
\end{pmatrix}\neq
\begin{pmatrix} 
\mathbf{1}_n & 0\\ 0 & \mathbf{1}_n	
\end{pmatrix}$, 
then $
\det\begin{pmatrix} 
C & BD\\ CA & B	
\end{pmatrix}=0
$.
\end{proof}

\subsection{Linear parametric Yang--Baxter maps}
\label{slpYBm}


Let $V$ be a vector space over a field $\mathbb{K}$.
Let $\Omega$ be a set.
Consider a family of linear maps $Y_{a,b}\in\End(V\times V)$ 
depending on parameters $a,b\in\Omega$.
One has
\begin{gather}
\label{YmapMatr}
\begin{gathered}
Y_{a,b}(x,y)=\big(u_{a,b}(x,y),v_{a,b}(x,y)\big),\qquad x,y\in V,\\
u_{a,b}(x,y)=A_{a,b}x+B_{a,b}y,\qquad
v_{a,b}(x,y)=C_{a,b}x+D_{a,b}y,
\end{gathered}
\end{gather}
for some linear maps $A_{a,b},B_{a,b},C_{a,b},D_{a,b}\in\End(V)$ depending on $a,b$.
As said in Section~\ref{sybm}, $Y_{a,b}$ is called a parametric YB map 
if it satisfies the parametric YB equation~\eqref{YB_eq}.

Let $\pyb(V)$ be the set of such linear parametric YB maps. 
An element $Y_{a,b}\in\pyb(V)$ given by~\eqref{YmapMatr} is written as 
$Y_{a,b}=(A_{a,b},D_{a,b};B_{a,b},C_{a,b})\in\pyb(V)$.

Substituting~\eqref{YmapMatr} in~\eqref{YB_eq}, 
we see that equation \eqref{YB_eq} is equivalent to the following system
\begin{subequations}
\label{analBraid}
\begin{gather}
\label{analBraid-a}
u_{a,b}(u_{a,c}(x,v_{b,c}(y,z)),u_{b,c}(y,z)) =u_{a,c}(u_{a,b}(x,y),z),\\
\label{analBraid-b}
v_{a,b}(u_{a,c}(x,v_{b,c}(y,z)),u_{b,c}(y,z))=u_{b,c}(v_{a,b}(x,y),v_{a,c}(u_{a,b}(x,y),z)),\\
\label{analBraid-c}
v_{a,c}(x,v_{b,c}(y,z))=v_{b,c}(v_{a,b}(x,y),v_{a,c}(u_{a,b}(x,y),z)),\qquad x,y,z\in V.
\end{gather}
\end{subequations}

Proposition~\ref{ABCD} below can be regarded as 
a generalisation of Proposition~\ref{pabcd} to the case of linear parametric YB maps.

\begin{proposition}
\label{ABCD}
A parametric family of maps $Y_{a,b}\in\End(V\times V)$ given by \eqref{YmapMatr}  
satisfies~\eqref{YB_eq} if and only if $A_{a,b}$, $B_{a,b}$, $C_{a,b}$, $D_{a,b}$ 
in \eqref{YmapMatr} obey the following relations for all values of the parameters $a,b,c\in\Omega$
\begin{subequations}
\begin{gather}
\label{ABC-relations-a}
C_{b, c}C_{a,b}=C_{a,c}-D_{b,c}C_{a,c}A_{a,b},\qquad 
B_{a, b}B_{b,c}=B_{a,c}-A_{a, b}B_{a,c}D_{b,c},\\
\label{ABC-relations-b}
D_{a,c}C_{b,c}-C_{b, c}D_{a,b}=D_{b, c}C_{a,c}B_{a,b},\qquad 
A_{a,c}B_{a,b}-B_{a, b}A_{b,c}=A_{a, b}B_{a,c}C_{b,c},\\
\label{ABC-relations-c}
C_{a, b}A_{a,c}-A_{b, c}C_{a,b}=B_{b, c}C_{a,c}A_{a,b},\qquad 
B_{b, c}D_{a,c}-D_{a, b}B_{b,c}=C_{a, b}B_{a,c}D_{b,c},\\
\label{ABC-relations-d}
D_{a, b}A_{b,c}-B_{b, c}C_{a,c}B_{a,b}=A_{b,c}D_{a,b}-C_{a, b}B_{a,c}C_{b,c},\\
\label{ABC-relations-e}
[A_{a, b }, A_{a, c}]=0,\qquad [D_{a, c}, D_{b,c}]=0.
\end{gather}
\end{subequations}

That is, for maps $A_{a,b}, B_{a,b}, C_{a,b}, D_{a,b}\in\End(V)$ depending on $a,b$, 
we have $(A_{a,b},D_{a,b};B_{a,b},C_{a,b})\in\pyb(V)$ if and only 
if $A_{a,b}$, $B_{a,b}$, $C_{a,b}$, $D_{a,b}$ 
obey relations \eqref{ABC-relations-a}--\eqref{ABC-relations-e}.
\end{proposition}
\begin{proof}
According to the left-hand side of \eqref{analBraid-a},
\begin{multline}
\label{ABCD-lhs}
u_{a,b}(u_{a,c}(x,v_{b,c}(y,z)),u_{b,c}(y,z))=A_{a,b}u_{a,c}(x,v_{b,c}(y,z))+B_{a,b}u_{b,c}(y,z)=\\
=A_{a,b}(A_{a,c}x+B_{a,c}v_{b,c}(y,z)) +B_{a,b}(A_{b,c}y+B_{b,c}z)=\\
=A_{a,b}(A_{a,c}x+B_{a,c}(C_{b,c}y+D_{b,c}z)) +B_{a,b}(A_{b,c}y+B_{b,c}z)=\\
=A_{a,b}A_{a,c}x+(A_{a,b}B_{a,c}C_{b,c}+B_{a,b}A_{b,c})y+(A_{a,b}B_{a,c}D_{b,c}+B_{a,b}B_{b,c})z.
\end{multline}
On the other hand, from the right-hand side of \eqref{analBraid-a}, we obtain
\begin{multline}
\label{ABCD-rhs}
u_{a,c}(u_{a,b}(x,y),z)=A_{a,c}u_{a,b}(x,y)+B_{a,c}z=\\
=A_{a,c}(A_{a,b}x+B_{a,b}y)+B_{a,c}z=A_{a,c}A_{a,b}x+A_{a,c}B_{a,b}y+B_{a,c}z. 
\end{multline}
Comparing the coefficients of $x$, $y$, $z$ in the right-hand sides of equations \eqref{ABCD-lhs} and \eqref{ABCD-rhs}, 
we deduce the first relation of \eqref{ABC-relations-e}, the second relation of \eqref{ABC-relations-b}, 
and the second relation of \eqref{ABC-relations-a}. 
Similarly, one can show that the rest of relations \eqref{ABC-relations-a}-\eqref{ABC-relations-e} 
are equivalent to \eqref{analBraid-b}-\eqref{analBraid-c}.
\end{proof}

\begin{theorem}
\label{thptr}
Let $V$ be a vector space over a field $\mathbb{K}$ and $\Omega$ be a set.
Consider a linear parametric YB map $Y_{a,b}=(A_{a,b},D_{a,b};B_{a,b},C_{a,b})\in\pyb(V)$
given by~\eqref{YmapMatr} with parameters $a,b\in\Omega$.

Then $(D_{b,a},A_{b,a};C_{b,a},B_{b,a})\in\pyb(V)$. That is, 
the set $\pyb(V)$ is invariant under the transformation
\begin{gather}
\label{pdacb}
(A_{a,b},D_{a,b};B_{a,b},C_{a,b})\mapsto (D_{b,a},A_{b,a};C_{b,a},B_{b,a}).
\end{gather}
This means that, if we set 
\begin{gather}
\label{taab}
\tilde{A}_{a,b}=D_{b,a},\qquad
\tilde{B}_{a,b}=C_{b,a},\qquad
\tilde{C}_{a,b}=B_{b,a},\qquad
\tilde{D}_{a,b}=A_{b,a},
\end{gather}
the maps $\tilde{A}_{a,b}$, $\tilde{B}_{a,b}$, $\tilde{C}_{a,b}$, $\tilde{D}_{a,b}$ 
obey relations \eqref{ABC-relations-a}--\eqref{ABC-relations-e}.

Furthermore, the set $\pyb(V)$ is invariant under the transformation
\begin{gather}
\label{pladl}
\begin{gathered}
(A_{a,b},D_{a,b};B_{a,b},C_{a,b})\mapsto (LA_{a,b},D_{a,b}L^{-1};B_{a,b},C_{a,b})\\
\text{for any invertible $L\in\End(V)$ that commutes with $A_{a,b}$, $B_{a,b}$, $C_{a,b}$, $D_{a,b}$ for all $a,b\in\Omega$}.
\end{gathered}
\end{gather}

Let $V=\mathbb{K}^n$ for some $n\in\zsp$. 
The space $\End(\mathbb{K}^n)$ is identified with $\mat_n(\mathbb{K})$.
The set $\pyb(\mathbb{K}^n)$ is invariant also under the transformations
\begin{gather}
\label{pttabcd}
(A_{a,b},D_{a,b};B_{a,b},C_{a,b})\mapsto(A^\ts_{a,b},D^\ts_{a,b};C^\ts_{a,b},B^\ts_{a,b}),\\
\label{ctrabcd}
(A_{a,b},D_{a,b};B_{a,b},C_{a,b})\mapsto(D^\ts_{b,a},A^\ts_{b,a};B^\ts_{b,a},C^\ts_{b,a}),
\end{gather}
where $\ts$ denotes the transpose operation.
\end{theorem}
\begin{proof}
According to Proposition~\ref{ABCD}, for $(A_{a,b},D_{a,b};B_{a,b},C_{a,b})\in\pyb(V)$
we have~\eqref{ABC-relations-a}--\eqref{ABC-relations-e}. 
To prove that $\pyb(V)$ is invariant under the transformation~\er{pdacb},
we need to show 
$(\tilde{A}_{a,b},\tilde{D}_{a,b};\tilde{B}_{a,b},\tilde{C}_{a,b})\in\pyb(V)$ for the maps~\er{taab}.

Since \eqref{ABC-relations-a} is valid for $(A_{a,b},D_{a,b};B_{a,b},C_{a,b})\in\pyb(V)$, 
using~\eqref{ABC-relations-a},~\er{taab}, we obtain
\begin{gather}
\lb{hfrab}
\tilde{B}_{c,b}\tilde{B}_{b,a}=\tilde{B}_{c,a}-\tilde{A}_{c,b}\tilde{B}_{c,a}\tilde{D}_{b,a},\qquad 
\tilde{C}_{b,a}\tilde{C}_{c,b}=\tilde{C}_{c,a}-\tilde{D}_{b,a}\tilde{C}_{c,a}\tilde{A}_{c,b}.
\end{gather}
Since equations~\eqref{hfrab} are valid for all values of $a,b,c$,
we can make the change $a\mapsto c,\ c\mapsto a$ in~\eqref{hfrab} and get 
\begin{gather}
\label{cdfrab}
\tilde{B}_{a,b}\tilde{B}_{b,c}=\tilde{B}_{a,c}-\tilde{A}_{a,b}\tilde{B}_{a,c}\tilde{D}_{b,c},\qquad 
\tilde{C}_{b,c}\tilde{C}_{a,b}=\tilde{C}_{a,c}-\tilde{D}_{b,c}\tilde{C}_{a,c}\tilde{A}_{a,b}.
\end{gather}
Equations~\eqref{cdfrab} say that 
$\tilde{A}_{a,b}$, $\tilde{B}_{a,b}$, $\tilde{C}_{a,b}$, $\tilde{D}_{a,b}$ 
satisfy relations~\eqref{ABC-relations-a}. 
In much the same way, 
one can show that $\tilde{A}_{a,b}$, $\tilde{B}_{a,b}$, $\tilde{C}_{a,b}$, $\tilde{D}_{a,b}$
given by~\eqref{taab} satisfy all relations~\eqref{ABC-relations-a}--\eqref{ABC-relations-e}, i.e., 
$(\tilde{A}_{a,b},\tilde{D}_{a,b};\tilde{B}_{a,b},\tilde{C}_{a,b})\in\pyb(V)$.

The other statements of the theorem are proved similarly.
Note that the transformation~\er{ctrabcd} is equal to 
the composition of the transformations~\er{pttabcd},~\er{pdacb}.
\end{proof}
\begin{corollary}
\lb{clc}
Let $(A_{a,b},D_{a,b};B_{a,b},C_{a,b})\in\pyb(V)$.
Then for any nonzero $l\in\mathbb{K}$ one has
\begin{gather}
\lb{lald}
(lA_{a,b},l^{-1}D_{a,b};B_{a,b},C_{a,b})\in\pyb(V).
\end{gather}
Thus we obtain a family of linear parametric YB maps~\er{lald} depending on~$l$.
\end{corollary}
\begin{proof}
Consider the identity map $\id\colon V\to V$.
Using~\er{pladl} for $L=l\cdot\id\in\End(V)$, one gets~\er{lald}.
\end{proof}

Since the values of the parameters $a,b,c\in\Omega$ are arbitrary, 
we are allowed to make any permutation of $a,b,c$ in equations~\eqref{ABC-relations-a}-\eqref{ABC-relations-e}.
Making the permutation $a\mapsto c$, $b\mapsto a$, $c\mapsto b$ 
in the first equation from~\eqref{ABC-relations-a} and in the first equation from~\eqref{ABC-relations-c}, we obtain
\begin{gather}
\label{cabcca}
C_{a,b}C_{c,a}=C_{c,b}-D_{a,b}C_{c,b}A_{c,a},\\
\label{ccaacb}
C_{c,a}A_{c,b}-A_{a,b}C_{c,a}=B_{a,b}C_{c,b}A_{c,a}.
\end{gather}

\begin{theorem}
\label{thme}
For any $Y_{a,b}=(A_{a,b},D_{a,b};B_{a,b},C_{a,b})\in\pyb(V)$ given by~\eqref{YmapMatr}, we have
\begin{gather}
\label{matrcdab}
\begin{pmatrix} 
C_{a,b} & D_{a,b}\\ A_{a,b} & B_{a,b}	
\end{pmatrix}
\begin{pmatrix} 
C_{c,a} & B_{a,c}D_{b,c}\\ C_{c,b}A_{c,a} & B_{b,c}	
\end{pmatrix}=
\begin{pmatrix} 
C_{c,b} & B_{b,c}D_{a,c}\\ C_{c,a}A_{c,b} & B_{a,c}	
\end{pmatrix},\\
\label{pmcbd}
\begin{pmatrix} 
C_{b,c} & D_{b,c}C_{a,c}\\ 
A_{c,a}B_{c,b} & B_{c,a}
\end{pmatrix}
\begin{pmatrix} 
C_{a,b} & D_{a,b}\\ A_{a,b} & B_{a,b}	
\end{pmatrix}
=\begin{pmatrix} 
C_{a,c} & D_{a,c}C_{b,c}\\ 
A_{c,b}B_{c,a} & B_{c,b}
\end{pmatrix}.
\end{gather}
Consider the map $P\in\End(V\times V)$, $P(x,y)=(y,x)$, and the maps
\begin{gather*}
H_{a,b,c}\in\End(V\times V),\qqquad H_{a,b,c}(x,y)=
(C_{c,a}x+B_{a,c}D_{b,c}y,\,C_{c,b}A_{c,a}x+B_{b,c}y),\qqquad x,y\in V,\\
\tilde{H}_{a,b,c}\in\End(V\times V),\qqquad\tilde{H}_{a,b,c}(x,y)=
(C_{b,c}x+D_{b,c}C_{a,c}y,\,A_{c,a}B_{c,b}x+B_{c,a}y),\qqquad x,y\in V,
\end{gather*}
depending on parameters $a,b,c\in\Omega$. Equations~\eqref{matrcdab},~\eqref{pmcbd} say that 
\begin{gather}
\label{ppyh}
PY_{a,b}H_{a,b,c}=H_{b,a,c},\\
\label{phpy}
\tilde{H}_{a,b,c}PY_{a,b}=\tilde{H}_{b,a,c}.
\end{gather}

\end{theorem}
\begin{proof}
Equation~\eqref{matrcdab} follows from~\eqref{ABC-relations-a}, 
\eqref{ABC-relations-c}, \eqref{cabcca}, \eqref{ccaacb}.
Equation~\eqref{pmcbd} is proved similarly.

Clearly, equations~\eqref{matrcdab},~\eqref{pmcbd} are equivalent to~\eqref{ppyh},~\eqref{phpy}.
\end{proof}

\begin{remark}
\label{rpequiv}
Equations~\eqref{matrcdab},~\eqref{pmcbd} are equivalent 
to~\eqref{ABC-relations-a},~\eqref{ABC-relations-b},~\eqref{ABC-relations-c},
up to permutations of $a,b,c$.
Thus, the rather cumbersome equations~\eqref{ABC-relations-a},~\eqref{ABC-relations-b},~\eqref{ABC-relations-c} 
can be replaced by equations~\eqref{matrcdab},~\eqref{pmcbd}, 
which have more clear structure, since they are of the form~\eqref{ppyh},~\eqref{phpy}.
\end{remark}

\begin{example}
\lb{exdmh}
Let $\mathbb{K}=\mathbb{C}$, $V=\mathbb{C}$, and $\Omega=\mathbb{C}^2$.
In~\cite{dim-mull} one can find the following linear parametric YB map
$Y_{a,b}\colon\mathbb{C}\times\mathbb{C}\to\mathbb{C}\times\mathbb{C}$
with $a=(a_1,a_2)\in\mathbb{C}^2$ and $b=(b_1,b_2)\in\mathbb{C}^2$
\begin{gather}
\lb{mhm}
Y_{a,b}\begin{pmatrix}x\\ y\end{pmatrix}=
\begin{pmatrix}
\frac{a_1-b_1}{a_1-b_2} & \frac{b_1-b_2}{a_1-b_2} \\[0.4em]
\frac{a_1-a_2}{a_1-b_2} & \frac{a_2-b_2}{a_1-b_2}
 \end{pmatrix}
\begin{pmatrix}x\\ y\end{pmatrix}, \qquad x,y\in\mathbb{C},\qquad a=(a_1,a_2),\quad
b=(b_1,b_2).
\end{gather}
In~\cite{dim-mull} the parameters $(a_1,a_2)$, $(b_1,b_2)$ are denoted by $(p_1,q_1)$, $(p_2,q_2)$.
Note that Remark~\ref{rrdp} is applicable to this map.

Let $l\in\mathbb{C}$, $l\neq 0$. Applying Corollary~\ref{clc} to 
$Y_{a,b}=
\big(\frac{a_1-b_1}{a_1-b_2},\frac{a_2-b_2}{a_1-b_2};\frac{b_1-b_2}{a_1-b_2},\frac{a_1-a_2}{a_1-b_2}\big)
\in\pyb(\mathbb{C})$, 
we obtain the linear parametric YB map
\begin{gather}
\lb{ylab}
Y^l_{a,b}\colon\mathbb{C}\times\mathbb{C}\to\mathbb{C}\times\mathbb{C},\qquad
Y^l_{a,b}\begin{pmatrix}x\\ y\end{pmatrix}=
\begin{pmatrix}
\frac{l(a_1-b_1)}{a_1-b_2} & \frac{b_1-b_2}{a_1-b_2} \\[0.4em]
\frac{a_1-a_2}{a_1-b_2} & \frac{a_2-b_2}{l(a_1-b_2)}
 \end{pmatrix}
\begin{pmatrix}x\\ y\end{pmatrix},\qquad
a=(a_1,a_2),\quad b=(b_1,b_2).
\end{gather}
For $l\neq 1$ the map~\er{ylab} is new.
\end{example}

\begin{remark}
\label{rreduce}
Theorem~\ref{thptr} can be regarded as 
a generalisation of Proposition~\ref{thtr} and Remark~\ref{rctr}
to the case of linear parametric YB maps.
However, if $Y_{a,b}=(A_{a,b},D_{a,b};B_{a,b},C_{a,b})\in\pyb(V)$ given by~\eqref{YmapMatr}
does not depend on the parameters $a,b\in\Omega$, then 
the requirements on $L\in\End(V)$ in~\er{ladl} are less strict than in~\er{pladl}, 
because in~\er{ladl} the map $L$ is not required to commute with $A$ and $D$.
(In~\er{ladl} the map $L$ commutes with $B$, $C$, and the product $AD$.)

Note also that, if $Y_{a,b}=(A_{a,b},D_{a,b};B_{a,b},C_{a,b})\in\pyb(V)$ 
does not depend on $a,b$, then equation~\er{matrcdab} reduces to~\er{bmcdab}.
\end{remark}

\section{Yang--Baxter maps from refactorisation problems}
\label{sybmr}

In this section, after recalling the well-known relation between matrix refactorisation problems 
and YB maps (see, e.g., \cite{Veselov,Veselov2,Kouloukas}),
in Proposition~\ref{pasa} we present a straightforward generalisation of this relation
to the case of refactorisation problems in arbitrary associative algebras.
Using Proposition~\ref{pasa}, we prove Proposition~\ref{paeps}, 
whose applications are discussed in Section~\ref{NLS-type-maps}.
Namely, Proposition~\ref{paeps} explains why the linear parametric maps 
arising from the linear approximations of the matrix refactorisation problems
considered in Section~\ref{NLS-type-maps} satisfy the parametric YB equation~\eqref{YB_eq}
(see Remarks~\ref{reps1},~\ref{reps2},~\ref{reps3}).

Let $V$, $\Omega$, $\Lambda$ be sets and $n\in\zsp$.
Let $L(x;a,\lambda)$ be an  $n\times n$ matrix depending 
on $x\in V$, $a\in\Omega$, $\lambda\in\Lambda$. 
Here $a$, $\lambda$ are regarded as parameters, 
and $\lambda$ is called a spectral parameter.
To simplify notation, we set $L_a(x)=L(x;a,\lambda)$, 
so $\lambda$ is not written explicitly in this notation.

Consider a family of maps 
\begin{gather}
\label{yabvv}
Y_{a,b}\colon V\times V\to V\times V,\qqquad
Y_{a,b}(x,y)=\big(u_{a,b}(x,y),\,v_{a,b}(x,y)\big),\qquad x,y\in V,
\end{gather}
depending on parameters $a,b\in\Omega$.
Suppose that $u=u_{a,b}(x,y)$ and $v=v_{a,b}(x,y)$ obey the equation
\begin{gather}
\label{eq-Lax}
L_a(u)L_b(v)=L_b(y)L_a(x)
\end{gather}
for all values of $x,y,a,b,\lambda$. 
Then, following~\cite{Veselov2}, we say that $L_a(x)=L(x;a,\lambda)$ 
is a \emph{Lax matrix} for the parametric map~\er{yabvv}.
Equation~\er{eq-Lax} is called the \emph{matrix refactorisation problem} 
corresponding to~$L_a(x)$.

Suppose that the equation
\begin{gather}
\label{trifac}
L_a(\hat{x})L_b(\hat{y})L_c(\hat{z})=L_a(x)L_b(y)L_c(z)\qquad\text{for all }\,a,b,c\in\Omega
\end{gather}
implies $\hat{x}=x$, $\hat{y}=y$, $\hat{z}=z$.
Then, if $L_a(x)$ is a Lax matrix for a parametric map~\er{yabvv},
this map satisfies the parametric YB equation \eqref{YB_eq} (see \cite{Veselov, Kouloukas}).

We need the following generalisation, where matrices are replaced by elements of an associative algebra.

\begin{proposition} 
\label{pasa}
Let $V$, $\Omega$ be sets and $\mathfrak{A}$ be an associative algebra.
Consider maps
\begin{gather}
\notag
Q_a\cl V\to\mathfrak{A},\qquad
u_{a,b}\cl V\times V\to V,\qquad
v_{a,b}\cl V\times V\to V
\end{gather}
depending on parameters $a,b\in\Omega$.
Suppose that
\begin{itemize}
\item  
the equation
\begin{gather}
\lb{qqq}
Q_a(\hat{x})Q_b(\hat{y})Q_c(\hat{z})=Q_a(x)Q_b(y)Q_c(z)\qquad\text{for all }\,a,b,c\in\Omega
\end{gather}
implies $\hat{x}=x$, $\hat{y}=y$, $\hat{z}=z$,
\item we have
\begin{gather}
\lb{quabx}
Q_a\big(u_{a,b}(x,y)\big)Q_b\big(v_{a,b}(x,y)\big)=Q_b(y)Q_a(x)
\qquad\text{for all }\,a,b\in\Omega,\,\ 
x,y\in V.
\end{gather}
\end{itemize}
Then the parametric map~\er{yabvv} satisfies the parametric YB equation \eqref{YB_eq}.
\end{proposition}
\begin{proof} 
In the case when $\mathfrak{A}$ is the algebra of $n\times n$ matrices 
depending on a spectral parameter, 
a proof of this statement is presented in~\cite{Kouloukas}.
The same proof works for arbitrary associative algebras.
\end{proof}
\begin{proposition} 
\label{paeps}
Let $V$, $\Omega$ be sets and $\mathfrak{B}$ be an associative algebra.
Consider maps
\begin{gather}
\notag
T_a\cl V\to\mathfrak{B},\qquad
u_{a,b}\cl V\times V\to V,\qquad
v_{a,b}\cl V\times V\to V
\end{gather}
depending on parameters $a,b\in\Omega$.
For each $a\in\Omega$, let $S_a\in\mathfrak{B}$ such that
\begin{gather}
\lb{sasb}
S_aS_b=S_bS_a\qquad\text{for all }\,a,b\in\Omega.
\end{gather}
Suppose that
\begin{itemize}
\item  
the equation
\begin{gather}
\lb{tss}
T_a(\hat{x})S_bS_c+S_aT_b(\hat{y})S_c+S_aS_bT_c(\hat{z})
=T_a(x)S_bS_c+S_aT_b(y)S_c+S_aS_bT_c(z)\qquad\text{for all }\,a,b,c\in\Omega
\end{gather}
implies $\hat{x}=x$, $\hat{y}=y$, $\hat{z}=z$,
\item we have
\begin{gather}
\label{satb}
S_aT_b\big(v_{a,b}(x,y)\big)+T_a\big(u_{a,b}(x,y)\big)S_b=
S_bT_a(x)+T_b(y)S_a
\qquad\text{for all }\,a,b\in\Omega,\,\ 
x,y\in V.
\end{gather}
\end{itemize}
Then the parametric map~\er{yabvv} satisfies the parametric YB equation \eqref{YB_eq}.
\end{proposition}
\begin{proof} 
Consider the formal symbol $\varepsilon$.
Let $\mathfrak{A}$ be the set of formal sums $m_1+\varepsilon m_2$, 
where $m_1,m_2\in\mathfrak{B}$. 
The set $\mathfrak{A}$ is an associative algebra with the following operations
\begin{gather}
\notag
\text{for all }\,m_1,m_2,\tilde{m}_1,\tilde{m}_2\in\mathfrak{B}\qquad
(m_1+\varepsilon m_2)+(\tilde{m}_1+\varepsilon\tilde{m}_2)=
(m_1+\tilde{m}_1)+\varepsilon(m_2+\tilde{m}_2),\\
\lb{mpm}
(m_1+\varepsilon m_2)(\tilde{m}_1+\varepsilon\tilde{m}_2)=
(m_1\tilde{m}_1)+\varepsilon(m_1\tilde{m}_2+m_2\tilde{m}_1).
\end{gather}
Note that, for any $m,\tilde{m}\in\mathfrak{B}$, one has
$(\varepsilon m)(\varepsilon\tilde{m})=0$ in $\mathfrak{A}$.
Thus essentially we assume $\varepsilon^2=0$.

For each $a\in\Omega$, consider the map
\begin{gather}
\lb{qava}
Q_a\cl V\to\mathfrak{A},\qqquad 
Q_a(w)=S_a+\varepsilon T_a(w),\qqquad w\in V.
\end{gather}
Using~\er{mpm},~\er{qava}, we obtain
\begin{gather*}
\text{for all }\,x,y\in V\qqquad
Q_b(y)Q_a(x)=S_bS_a+\varepsilon\big(S_bT_a(x)+T_b(y)S_a\big),\\
Q_a\big(u_{a,b}(x,y)\big)Q_b\big(v_{a,b}(x,y)\big)=
S_aS_b+\varepsilon\big(S_aT_b\big(v_{a,b}(x,y)\big)+T_a\big(u_{a,b}(x,y)\big)S_b\big).
\end{gather*}
Hence equations~\er{sasb},~\er{satb} are equivalent to~\er{quabx}.
Similarly, equation~\er{tss} is equivalent to~\er{qqq}.
Therefore, we can use Proposition~\ref{pasa}.
\end{proof}

\section{Linear parametric Yang--Baxter maps related to Darboux transformations for NLS-type equations}
\label{NLS-type-maps}

It is well known that YB maps are closely related to quad-graph equations (namely, partial difference equations defined on an elementary square of the two-dimensional lattice),
see, e.g.,~\cite{abs2003,hjn-book,PTV} and references therein. 
Recall that Darboux and B\"acklund transformations can be employed in order to derive quad-graph equations~\cite{Adler,hjn-book,SPS,Nijhoff,QNC}. In particular, the Bianchi permutability of Darboux transformations yields integrable partial difference equations. Since the former permutability condition of Darboux transformations can be regarded 
as a matrix refactorisation problem similar to those described in Section~\ref{sybmr}, 
this suggests to consider matrix refactorisation problems for particular Darboux matrices 
in order to construct YB maps. 

In \cite{Sokor-Sasha, KP2019}, Darboux matrix refactorisation problems related 
to NLS-type partial differential equations were considered, and new birational parametric YB maps were constructed.
In this section, developing an idea from~\cite{Sokor-Sasha}, 
we demonstrate how to obtain linear parametric YB maps
(with nonlinear dependence on parameters), 
using linear approximations of such matrix refactorisation problems.
Relations of our results with those of~\cite{Sokor-Sasha} 
are discussed in Remark~\ref{rmikh} and in Subsection~\ref{saddnls}.

\subsection{Darboux transformations for NLS-type equations}
The Darboux transformations which were employed in \cite{Sokor-Sasha} are associated with AKNS-type Lax operators of the form $\mathcal{L}=D_x+U$, where $U=U(p,q;\lambda)$ 
is a $2\times 2$ matrix-function with zero trace.
Here $p=p(x,t)$ and $q=q(x,t)$ are potential functions, which are solutions to a certain NLS-type equation,
and $\lambda\in\mathbb{C}$ is a parameter called the spectral parameter. 

In the considered Lax operators, $U$ depends rationally on $\lambda$, 
so $U$ can be viewed as a function of $p,q$ 
with values in the Lie algebra $\mathfrak{sl}_2(\mathbb{C}(\lambda))$, 
where $\mathbb{C}(\lambda)$ is the commutative algebra of rational functions of~$\lambda$. 

Following~\cite{Sokor-Sasha,SPS}, we say that, in this case, a \emph{Darboux transformation}
is determined by an invertible $2\times 2$ matrix $M$ (called a \emph{Darboux matrix}) such that
\begin{gather}
\lb{lpq}
M \mathcal{L} M^{-1}=M\big(D_x+U(p,q;\lambda)\big)M^{-1}=D_x-D_x(M)M^{-1}+MUM^{-1}=
D_x+U(\tilde{p},\tilde{q};\lambda),
\end{gather}
where functions $\tilde{p},\tilde{q}$ are also solutions of the same NLS-type equation. 
The matrix $M$ may depend on the potential functions $p,q,\tilde{p},\tilde{q}$,
the parameter $\lambda$, and some auxiliary functions.

Since the matrix $U$ has zero trace, from~\er{lpq} we see that the trace of $D_x(M)M^{-1}$ is zero, 
which yields $D_x(\det(M))=0$. This condition implies some relations between 
the potential functions $p,q,\tilde{p},\tilde{q}$ and the auxiliary functions that appear in the derivation of the Darboux matrix. 
Usually these relations allow one to express the auxiliary functions in terms of the potential functions. 
Depending on the complexity of the relations, sometimes one cannot derive a map from the Darboux matrix refactorisation problem, but a correspondence. However, in such cases, one can often derive a linear approximation to the map. This will be clear in the following subsections where we construct examples of linear parametric YB maps associated with Darboux transformations for the derivative NLS (DNLS) equation 
and a deformation of the derivative NLS (DDNLS) equation which first appeared in~\cite{MSY}.

All the Darboux transformations that are being used in the next subsections 
were constructed in~\cite{SPS} and are associated with the following Lax operators
\begin{align*}
\mathcal{L}_{DNLS}&=D_x+\lambda^2 \begin{pmatrix}1 & 0\\0 & -1\end{pmatrix}
+\lambda\begin{pmatrix}0 & 2p\\2q & 0\end{pmatrix},\\
\mathcal{L}_{DDNLS}&=D_x+\lambda^2 \begin{pmatrix}1 & 0\\0 & -1\end{pmatrix}
+\lambda\begin{pmatrix}0 & 2p\\2q & 0\end{pmatrix}
+\frac{1}{\lambda}\begin{pmatrix}0 & 2q\\2p & 0\end{pmatrix}
-\frac{1}{\lambda^2}\begin{pmatrix}1 & 0\\0 & -1\end{pmatrix},
\end{align*}
i.e., the spatial parts of the Lax pairs for the DNLS and DDNLS equations.

\subsection{A linear parametric Yang--Baxter map in the DNLS case}
\lb{sdnlsyb}

A Darboux matrix associated to the DNLS equation is the following
\begin{gather}\label{DNLS}
M=f\left(
\lambda^2
\begin{pmatrix}
 1 & 0\\
 0 & 0
\end{pmatrix}
+\lambda
\begin{pmatrix}
 0 & p\\
 \tilde{q} & 0
\end{pmatrix}
\right)
+
\begin{pmatrix}
1 & 0\\
 0 & 1
\end{pmatrix},
\end{gather}
where $p,\tilde{q}$ are potential functions, solutions to the DNLS equation, 
and $f$ is an arbitrary function which appears in the derivation of the Darboux matrix $M$. 
Moreover, $p,\tilde{q},f$ satisfy a system of differential-difference equations \cite{SPS}, 
which admits the following first integral
\begin{gather}
\lb{phif}
\Phi=f^2p\tilde{q}-f.
\end{gather}
That is, $D_x(\Phi)=0$, which is equivalent to $D_x(\det(M))=0$. 

Therefore, we can impose the relation $\Phi=\text{const}$, 
which allows us to determine the function~$f$. 
We set $\Phi=-a$, where $a\in\mathbb{C}$,
and replace $(p,\tilde{q})\rightarrow (\varepsilon x_1,\varepsilon x_2)$.
Then~\er{phif} becomes 
\begin{gather}
\lb{af2}
-a=f^2x_1x_2\varepsilon^2-f.
\end{gather}
Expanding in $\varepsilon$ around 0, we consider equation~\er{af2} up to $\mathcal{O}(\varepsilon^2)$
and take $f=a+\mathcal{O}(\varepsilon^2)$.

The matrix $M$ in \eqref{DNLS} now becomes
\begin{gather}\label{DNLS-2}
M(x_1,x_2,a)=
\lambda^2
\begin{pmatrix}
 a & 0\\
 0 & 0
\end{pmatrix}
+
\begin{pmatrix}
1 & 0\\
 0 & 1
\end{pmatrix}
+\varepsilon\lambda
\begin{pmatrix}
 0 & a x_1\\
 ax_2 & 0
\end{pmatrix}
+\mathcal{O}(\varepsilon^2).
\end{gather}
For the matrix $M_a(x_1,x_2)\equiv M(x_1,x_2,a)$ in \eqref{DNLS-2},
consider the matrix refactorisation problem
\begin{gather}
\label{mmmm}
M_a(u_1,u_2)M_b(v_1,v_2)=M_b(y_1,y_2)M_a(x_1,x_2)
\quad\text{up to $\mathcal{O}(\varepsilon^2)$}.
\end{gather}
After expanding \eqref{mmmm} in $\varepsilon$, we compute 
the coefficients of~$\varepsilon$ (i.e., the terms of degree~$1$ in~$\varepsilon$), 
which give the system of equations
\begin{gather}
\lb{linuv}
au_1+bv_1=ax_1+by_1,\qquad v_1=x_1,\qquad u_2=y_2,\qquad au_2+bv_2=ax_2+by_2.
\end{gather}
It is easy to check that the terms of degree~$0$ in~$\varepsilon$ in equation~\er{mmmm} cancel.
Therefore, equation~\er{mmmm} (considered up to $\mathcal{O}(\varepsilon^2)$) is equivalent to~\er{linuv}.

System~\er{linuv} can uniquely be solved for $u_1,u_2,v_1,v_2$, which gives the following map
\begin{gather}\label{Z2-YB}
Y_{a,b}\left(\begin{matrix}x_1\\ x_2\\ y_1\\y_2	\end{matrix}\right)=
\left(\begin{matrix} u_1\\ u_2 \\ v_1\\ v_2	\end{matrix}\right)_{a,b}\equiv
\begin{pmatrix}
1-\frac{b}{a} & 0 & \frac{b}{a} & 0 \\
 0 & 0 & 0 & 1 \\
 1 & 0 & 0 & 0 \\
0 & \frac{a}{b} & 0 & 1-\frac{a}{b} \\
\end{pmatrix}
\left(\begin{matrix}x_1\\ x_2\\ y_1\\y_2	\end{matrix}\right).
\end{gather}

It is easy to check that the corresponding matrices 
\begin{eqnarray*}
&A_{a,b}=\left(\begin{matrix}1-\frac{b}{a} & 0\\ 0 & 0	\end{matrix}\right),\qquad
B_{a,b}=\left(\begin{matrix} \frac{b}{a} & 0\\ 0 & 1	\end{matrix}\right),&\\
&C_{a,b}=\left(\begin{matrix}1 & 0 \\ 0 & \frac{b}{a}	\end{matrix}\right),\qquad 
D_{a,b}=\left(\begin{matrix}0 & 0\\  0 & 1-\frac{a}{b}	\end{matrix}\right)&
\end{eqnarray*}
obey relations \eqref{ABC-relations-a}--\eqref{ABC-relations-e}. 
Therefore, \eqref{Z2-YB} is a linear parametric YB map.
In Remark~\ref{reps1} we explain this by means of Proposition~\ref{paeps}.
\begin{remark}
\lb{reps1}
Let $V=\mathbb{C}^2$, $\Omega=\mathbb{C}$, and $\mathfrak{B}=\mat_2(\mathbb{C}[\lambda])$.
For $a\in\Omega=\mathbb{C}$, one can rewrite formula~\er{DNLS-2} as
$M(x_1,x_2,a)=S_a+\varepsilon T_a(x_1,x_2)+\mathcal{O}(\varepsilon^2)$, where
\begin{gather}
\lb{satar}
S_a=\lambda^2\begin{pmatrix}
 a & 0\\
 0 & 0
\end{pmatrix}
+ \begin{pmatrix}
1 & 0\\
 0 & 1
\end{pmatrix}\in\mathfrak{B},\qquad
T_a\cl\mathbb{C}^2\to\mathfrak{B},\,\quad
T_a(x_1,x_2)=\lambda\begin{pmatrix}
 0 & a x_1\\
 ax_2 & 0
\end{pmatrix},\,\quad x_1,x_2\in\mathbb{C}.
\end{gather}
Since we make computations up to $\mathcal{O}(\varepsilon^2)$, 
we can use formula~\er{mpm}.
Proposition~\ref{paeps} is applicable here, which explains why the map~\er{Z2-YB}
obtained from~\er{mmmm}, where $M_a(x_1,x_2)=M(x_1,x_2,a)$,
satisfies the parametric YB equation~\eqref{YB_eq}.

Equation~\er{mmmm} is studied up to $\mathcal{O}(\varepsilon^2)$ 
and is equivalent to linear equations~\er{linuv}.
This allows one to say that~\er{mmmm} can be regarded as a linear approximation 
of the matrix refactorisation problem corresponding to the matrix-function~\er{DNLS}.

Recall that the formula $M(x_1,x_2,a)=S_a+\varepsilon T_a(x_1,x_2)+\mathcal{O}(\varepsilon^2)$
with $S_a$, $T_a$ given by~\er{satar}
is obtained from~\er{DNLS}, using the substitution 
$(p,\tilde{q})\rightarrow (\varepsilon x_1,\varepsilon x_2)$.
We use also the formula $f=a+\mathcal{O}(\varepsilon^2)$ suggested by equation~\er{af2},
which is obtained from the equation $\Phi=-a$ by this substitution.
The substitution is chosen so that the resulting $S_a$, $T_a$ 
obey the conditions of Proposition~\ref{paeps}.
In particular, since the matrix $S_a$ in~\er{satar} is diagonal, we have~\er{sasb}.
\end{remark}

The map~\eqref{Z2-YB} can be derived from the map~\er{mhm} as follows.
Substituting $a_2=b_2=0$ in~\er{mhm} and denoting $a=a_1$, $b=b_1$, we obtain the map
\begin{gather}
\lb{mhm1}
\tilde{Y}_{a,b}\colon\mathbb{C}\times\mathbb{C}\to\mathbb{C}\times\mathbb{C},\qqquad
\tilde{Y}_{a,b}\begin{pmatrix}x\\ y\end{pmatrix}=
\begin{pmatrix}
\frac{a-b}{a} & \frac{b}{a} \\[0.4em]
1 & 0
 \end{pmatrix}
\begin{pmatrix}x\\ y\end{pmatrix},\qqquad x,y\in\mathbb{C}.
\end{gather}
Substituting $a_1=b_1=0$ in~\er{mhm} and denoting $a=a_2$, $b=b_2$, we obtain the map
\begin{gather}
\lb{mhm2}
\hat{Y}_{a,b}\colon\mathbb{C}\times\mathbb{C}\to\mathbb{C}\times\mathbb{C},\qqquad
\hat{Y}_{a,b}\begin{pmatrix}x\\ y\end{pmatrix}=
\begin{pmatrix}
0 & 1 \\[0.4em]
\frac{a}{b} & \frac{b-a}{b}
 \end{pmatrix}
\begin{pmatrix}x\\ y\end{pmatrix},\qqquad x,y\in\mathbb{C}.
\end{gather}
The map~\eqref{Z2-YB} is equal to the direct sum of the maps~\er{mhm1},~\er{mhm2}.

\subsection{The DDNLS case}
\lb{sbDDNLS}
A Darboux matrix associated to the DDNLS equation is given by
\begin{gather}\label{DDNLS}
M=
f
\begin{pmatrix}
 \lambda^2 & 0\\
 0 & \lambda^{-2}
\end{pmatrix}
+\lambda
\begin{pmatrix}
 0 & fp\\
 f\tilde{q} & 0
\end{pmatrix}
+fg
\begin{pmatrix}
 1 & 0\\
 0 & 1
\end{pmatrix}
+\frac{1}{\lambda}
\begin{pmatrix}
 0 & f\tilde{q}\\
 fp & 0
\end{pmatrix},
\end{gather}
where $p,\tilde{q},f,g$ satisfy a system of differential-difference equations~\cite{SPS},
which admits two first integrals $\Phi_1$, $\Phi_2$
\begin{gather}\label{first_ints}
\Phi_1=f^2(g-p\tilde{q}),\quad\qquad \Phi_2=f^2(g^2+1-p^2-\tilde{q}^2),
\end{gather}
i.e., $D_x(\Phi_i)=0$, $i=1,2$. The latter guarantees that $D_x(\det(M))=0$.

Hence we can impose the relations $\Phi_1=c_1$ and $\Phi_2=c_2$ 
for constants $c_1,c_2\in\mathbb{C}$, which allows us to determine the functions $f$, $g$.
As noticed in~\cite{Sokor-Sasha},
it is convenient to consider the case $c_1=\frac{1-k^2}{4}$, $c_2=\frac{1+k^2}{2}$ 
for a constant $k\in\mathbb{C}$. 
In the obtained relations $\Phi_1=\frac{1-k^2}{4}$, $\Phi_2=\frac{1+k^2}{2}$ and in~\er{DDNLS}
we replace $(p,\tilde{q})\rightarrow (\varepsilon x_1,\varepsilon x_2)$, which gives
\begin{gather}
\label{eqkeps}
f^2(g-x_1x_2\varepsilon^2)=\frac{1-k^2}{4},\quad\qquad 
f^2(g^2+1-x_1^2\varepsilon^2-x_2^2\varepsilon^2)=\frac{1+k^2}{2},\\
\label{epsDDNLS}
M=
f
\begin{pmatrix}
 \lambda^2 & 0\\
 0 & \lambda^{-2}
\end{pmatrix}
+\lambda
\begin{pmatrix}
 0 & f\varepsilon x_1\\
 f\varepsilon x_2& 0
\end{pmatrix}
+fg
\begin{pmatrix}
 1 & 0\\
 0 & 1
\end{pmatrix}
+\frac{1}{\lambda}
\begin{pmatrix}
 0 & f\varepsilon x_2\\
 f\varepsilon x_1& 0
\end{pmatrix}.
\end{gather}
Expanding in $\varepsilon$ around 0, 
we solve equations~\er{eqkeps} for $f$, $g$ up to $\mathcal{O}(\varepsilon^2)$
and obtain $4$ cases
\begin{gather}
\lb{f1k}
f=\frac{1-k}{2}+\mathcal{O}(\varepsilon^2),\quad\qquad g=\frac{1+k}{1-k}+\mathcal{O}(\varepsilon^2),\\
\lb{fk1}
f=\frac{k-1}{2}+\mathcal{O}(\varepsilon^2),\quad\qquad g=\frac{1+k}{1-k}+\mathcal{O}(\varepsilon^2),\\
\lb{f1pk}
f=\frac{1+k}{2}+\mathcal{O}(\varepsilon^2),\quad\qquad g=\frac{1-k}{1+k}+\mathcal{O}(\varepsilon^2),\\
\lb{mf1pk}
f=-\frac{1+k}{2}+\mathcal{O}(\varepsilon^2),\quad\qquad g=\frac{1-k}{1+k}+\mathcal{O}(\varepsilon^2).
\end{gather}

Let us study first the case~\er{f1k}.
Substituting~\er{f1k} in~\er{epsDDNLS} and denoting 
the obtained matrix by $M_k(x_1,x_2)$, one derives the formula
\begin{multline}\label{e-DDNLS}
M_k(x_1,x_2)=
\frac{1-k}{2}
\begin{pmatrix}
 \lambda^2 & 0\\
 0 & \lambda^{-2}
\end{pmatrix}
+\frac{1+k}{2}
\begin{pmatrix}
 1 & 0\\
 0 & 1
\end{pmatrix}
+
\varepsilon\frac{1-k}{2}
\begin{pmatrix}
 0 & \lambda x_1+\lambda^{-1} x_2\\
\lambda x_2+\lambda^{-1} x_1 & 0
\end{pmatrix}
+\mathcal{O}(\varepsilon^2).
\end{multline}
For the matrix $M_k(x_1,x_2)$ in \eqref{e-DDNLS}, 
consider the following matrix refactorisation problem
\begin{gather}
\lb{emmab}
M_a(u_1,u_2)M_b(v_1,v_2)=M_b(y_1,y_2)M_a(x_1,x_2)
\quad\text{up to $\mathcal{O}(\varepsilon^2)$}.
\end{gather}
After expanding \eqref{emmab} in $\varepsilon$, 
we consider the coefficients of $\varepsilon$, which give the system of equations
\begin{gather}
\lb{luyx}
u_2=y_2, \qqquad v_1=x_1,\\
\lb{labuvy}
\begin{split}
(a-1)(b-1)u_1-&(a-1)(b+1)u_2-(a+1)(b-1)v_2=\\
&=(a-1)(b-1)y_1-(a+1)(b-1)y_2-(b-1)(a+1)x_2,
\end{split}\\
\lb{labuva}
\begin{split}
(a-1)(b+1)u_1+&(a+1)(b-1)v_1-(a-1)(b-1)v_2=\\
&=(a+1)(b-1)y_1+(a-1)(b+1)x_1-(a-1)(b-1)x_2.
\end{split}
\end{gather}
The above system can uniquely be solved for $u_1,u_2,v_1,v_2$, which gives the map
\begin{gather}
\label{D2-YB}
Y_{a,b}\left(\begin{matrix}x_1\\ x_2\\ y_1\\y_2	\end{matrix}\right)=
\left(\begin{matrix} u_1\\ u_2 \\ v_1\\ v_2	\end{matrix}\right)_{a,b}\equiv
\begin{pmatrix}
 \frac{(a+1)(a-b)}{(a-1)(a+b)} & \frac{a-b}{a+b} & \frac{2 a (b-1)}{(a-1)(a+b)} & -\frac{a-b}{a+b} \\
 0 & 0 & 0 & 1 \\
 1 & 0 & 0 & 0 \\
 \frac{a-b}{a+b} & \frac{2 (a-1) b}{(b-1)(a+b)} & -\frac{a-b}{a+b} & -\frac{(a-b)(b+1)}{(b-1)(a+b)} \\
\end{pmatrix}
\left(\begin{matrix}x_1\\ x_2\\ y_1\\y_2	\end{matrix}\right).
\end{gather}

The matrix representing the linear map \eqref{D2-YB} satisfies 
the relations of Proposition~\ref{ABCD}. Namely, the corresponding matrices 
\begin{align*}
A_{a,b}&=\left(\begin{matrix} \frac{(a+1)(a-b)}{(a-1)(a+b)} & \frac{a-b}{a+b}\\ 0 & 0	\end{matrix}\right),\qquad
B_{a,b}=\left(\begin{matrix}  \frac{2 a (b-1)}{(a-1)(a+b)} & -\frac{a-b}{a+b}\\ 0 & 1	\end{matrix}\right),&\\
C_{a,b}&=\left(\begin{matrix} 1 & 0 \\ \frac{a-b}{a+b} &  \frac{2 (a-1) b}{(b-1)(a+b)} 	\end{matrix}\right),\qquad 
D_{a,b}=\left(\begin{matrix}0 & 0\\  -\frac{a-b}{a+b} & -\frac{(a-b)(b+1)}{(b-1)(a+b)} 	\end{matrix}\right)&
\end{align*}
obey relations \eqref{ABC-relations-a}--\eqref{ABC-relations-e}.
Hence \eqref{D2-YB} is a linear parametric YB map.
As shown in Remark~\ref{reqybm}, 
the map~\eqref{D2-YB} is equivalent to the YB map~\er{ybmss}, 
which appeared in \cite{Sokor-Sasha}.

\begin{remark}
\lb{reps2}
Similarly to Remark~\ref{reps1},
the fact that the map~\eqref{D2-YB} satisfies the parametric YB equation~\eqref{YB_eq}
is explained by Proposition~\ref{paeps}.

Equation~\er{emmab} is studied up to $\mathcal{O}(\varepsilon^2)$ 
and is equivalent to linear equations~\er{luyx},~\er{labuvy},~\er{labuva}.
This allows one to say that~\er{emmab} can be regarded as a linear approximation 
of the matrix refactorisation problem corresponding to the matrix-function~\er{DDNLS}.
\end{remark}

So in the case~\er{f1k}, using equation~\er{emmab}, we have obtained the map~\er{D2-YB}.
The case~\er{fk1} is obtained from~\er{f1k} by the change $f\mapsto -f$.
Since the matrix $M$ in~\er{epsDDNLS} is of the form $M=fN(x_1,x_2,g,\lambda,\varepsilon)$ 
for some matrix $N(x_1,x_2,g,\lambda,\varepsilon)$, 
the change $f\mapsto -f$ does not affect equation~\er{emmab}, 
hence the case~\er{fk1} gives the same map~\er{D2-YB}.

The case~\er{f1pk} is obtained from~\er{f1k} by the change $k\mapsto -k$, 
hence in this case we need to make the change $k\mapsto -k$ in the right-hand side of~\er{e-DDNLS}.
Then the above procedure gives the map~\er{D2-YB} with $a,b$ replaced by $-a,-b$.

The case~\er{mf1pk} is obtained from~\er{f1k} by the changes $f\mapsto -f$ and $k\mapsto -k$.
By the above arguments, in this case we get the map~\er{D2-YB} with $a,b$ replaced by $-a,-b$.

\subsection{More results on the DDNLS case}
\lb{saddnls}
Following~\cite{Sokor-Sasha}, consider again the relations
$\Phi_1=\frac{1-k^2}{4}$, $\Phi_2=\frac{1+k^2}{2}$ for a constant $k\in\mathbb{C}$,
where $\Phi_i$, $i=1,2$, are given by~\eqref{first_ints}.
In the relations $\Phi_1=\frac{1-k^2}{4}$, $\Phi_2=\frac{1+k^2}{2}$ and in~\er{DDNLS}
we replace $(fp,f\tilde{q})\rightarrow (\varepsilon x_1,\varepsilon x_2)$, which gives
\begin{gather}
\label{neqkeps}
f^2g-x_1x_2\varepsilon^2=\frac{1-k^2}{4},\qqquad
f^2g^2+f^2-x_1^2\varepsilon^2-x_2^2\varepsilon^2=\frac{1+k^2}{2},\\
\label{nepsDDNLS}
M=
f
\begin{pmatrix}
 \lambda^2 & 0\\
 0 & \lambda^{-2}
\end{pmatrix}
+\lambda
\begin{pmatrix}
 0 & \varepsilon x_1\\
 \varepsilon x_2& 0
\end{pmatrix}
+fg
\begin{pmatrix}
 1 & 0\\
 0 & 1
\end{pmatrix}
+\frac{1}{\lambda}
\begin{pmatrix}
 0 & \varepsilon x_2\\
 \varepsilon x_1& 0
\end{pmatrix}.
\end{gather}
Expanding in $\varepsilon$ around 0, 
we now solve equations~\er{neqkeps} for $f$, $fg$ up to $\mathcal{O}(\varepsilon^2)$
and obtain $4$ cases
\begin{gather}
\lb{1pk2}
f=\frac{1+k}{2}+\mathcal{O}(\varepsilon^2),\quad\qquad fg=\frac{1-k}{2}+\mathcal{O}(\varepsilon^2),\\
\lb{m1pk2}
f=-\frac{1+k}{2}+\mathcal{O}(\varepsilon^2),\quad\qquad fg=\frac{k-1}{2}+\mathcal{O}(\varepsilon^2),\\
\lb{1mk2}
f=\frac{1-k}{2}+\mathcal{O}(\varepsilon^2),\quad\qquad fg=\frac{1+k}{2}+\mathcal{O}(\varepsilon^2),\\
\lb{km12}
f=\frac{k-1}{2}+\mathcal{O}(\varepsilon^2),\quad\qquad fg=-\frac{1+k}{2}+\mathcal{O}(\varepsilon^2).
\end{gather}

Consider the case~\er{1pk2}. Substituting~\er{1pk2} in~\er{nepsDDNLS} 
and denoting the obtained matrix by $M_k(x_1,x_2)$, we derive the formula
\begin{gather}
\label{mform}
M_k(x_1,x_2)=
\frac{1+k}{2}
\begin{pmatrix}
 \lambda^2 & 0\\
 0 & \lambda^{-2}
\end{pmatrix}
+\frac{1-k}{2}
\begin{pmatrix}
 1 & 0\\
 0 & 1
\end{pmatrix}
+\varepsilon
\begin{pmatrix}
 0 &  \lambda x_1+\lambda^{-1} x_2\\
  \lambda x_2+\lambda^{-1} x_1& 0
\end{pmatrix}
+\mathcal{O}(\varepsilon^2).
\end{gather}
By the same procedure as in Subsection~\ref{sbDDNLS}, 
considering equation~\er{emmab} for the matrix~\er{mform}, 
one obtains the following linear parametric YB map
\begin{gather}
\lb{ybmss}
Y_{a,b}\left(\begin{matrix}x_1\\ x_2\\ y_1\\y_2	\end{matrix}\right)=
\left(\begin{matrix} u_1\\ u_2 \\ v_1\\ v_2	\end{matrix}\right)_{a,b}\equiv
\begin{pmatrix}
 \frac{(a-1) (a-b)}{(a+1) (a+b)} & \frac{a-b}{a+b} & \frac{2 a}{a+b} & -\frac{(a+1) (a-b)}{(b+1) (a+b)} \\
 0 & 0 & 0 & \frac{a+1}{b+1} \\
 \frac{b+1}{a+1} & 0 & 0 & 0 \\
 \frac{(a-b) (b+1)}{(a+1) (a+b)} & \frac{2 b}{a+b} & -\frac{a-b}{a+b} & -\frac{(a-b) (b-1)}{(b+1) (a+b)} \\
\end{pmatrix}
\left(\begin{matrix}x_1\\ x_2\\ y_1\\y_2	\end{matrix}\right),
\end{gather}
which appeared in \cite{Sokor-Sasha}.

\begin{remark}
\lb{reps3}
As said above, the linear map~\er{ybmss} is obtained 
from equation~\er{emmab} for~\er{mform}.
Similarly to Remarks~\ref{reps1},~\ref{reps2},
the fact that the map~\er{ybmss} satisfies the parametric YB equation~\eqref{YB_eq}
is explained by Proposition~\ref{paeps}.
\end{remark}

In the case~\er{m1pk2} one obtains the same map.
In the cases~\er{1mk2},~\er{km12} one derives
the map~\er{ybmss} with $a,b$ replaced by $-a,-b$.

Let $l\in\mathbb{C}$, $l\neq 0$. 
Applying Corollary~\ref{clc} to the map~\eqref{ybmss}, we get the map
\begin{gather}
\lb{lybmss}
Y^l_{a,b}\colon\mathbb{C}^2\times\mathbb{C}^2\to\mathbb{C}^2\times\mathbb{C}^2,\qquad
Y^l_{a,b}\left(\begin{matrix}x_1\\ x_2\\ y_1\\y_2	\end{matrix}\right)=
\begin{pmatrix}
 \frac{l(a-1) (a-b)}{(a+1) (a+b)} & \frac{l(a-b)}{a+b} & \frac{2 a}{a+b} & -\frac{(a+1) (a-b)}{(b+1) (a+b)} \\
 0 & 0 & 0 & \frac{a+1}{b+1} \\
 \frac{b+1}{a+1} & 0 & 0 & 0 \\
 \frac{(a-b) (b+1)}{(a+1) (a+b)} & \frac{2 b}{a+b} & -\frac{a-b}{l(a+b)} & -\frac{(a-b) (b-1)}{l(b+1) (a+b)} \\
\end{pmatrix}
\left(\begin{matrix}x_1\\ x_2\\ y_1\\y_2	\end{matrix}\right).
\end{gather}
For $l\neq 1$ the YB map~\er{lybmss} is new.

\begin{remark}
\lb{reqybm}
The following observation was brought to us by A.V.~Mikhailov.
Consider the matrix
$$
M_{a,b}=
\begin{pmatrix}
 \frac{(a+1)(a-b)}{(a-1)(a+b)} & \frac{a-b}{a+b} & \frac{2 a (b-1)}{(a-1)(a+b)} & -\frac{a-b}{a+b} \\
 0 & 0 & 0 & 1 \\
 1 & 0 & 0 & 0 \\
 \frac{a-b}{a+b} & \frac{2 (a-1) b}{(b-1)(a+b)} & -\frac{a-b}{a+b} & -\frac{(a-b)(b+1)}{(b-1)(a+b)} \\
\end{pmatrix}
$$
from~\er{D2-YB} and the diagonal matrix $U_{a,b}=\mathrm{diag}(a-1,a-1,b-1,b-1)$. 
The matrix $U_{a,b}M_{a,b}U_{a,b}^{-1}$ is equal to the $4\times 4$ matrix in~\er{ybmss} 
after the change $a\mapsto -a$, $\,b\mapsto -b$.
This implies that the linear parametric YB maps~\er{D2-YB} and~\er{ybmss}
are equivalent up to a change of the basis in~$\mathbb{C}^2\times\mathbb{C}^2$ 
and the change $a\mapsto -a$, $b\mapsto -b$.
\end{remark}

\section{Yang--Baxter maps associated with matrix groups}
\label{sYBmvb}

Let $\md$ be a group. One has the following YB map~\cite{Drin92}
\begin{gather}
\lb{fxyx}
\fmp\cl \md\times \md\to \md\times \md,\qquad
\fmp(x,y)=(x,xyx^{-1}),\qquad
x,y\in \md.
\end{gather}

Assume that $\mathbb{K}$ is either $\mathbb{R}$ or $\mathbb{C}$.
Let $n\in\zsp$ and consider the group
$\md=\GL_n(\mathbb{K})\subset\mat_n(\mathbb{K})$.
Then $\md$ is a manifold, and for each $x\in\md=\GL_n(\mathbb{K})$ 
one has the tangent space $T_x\md\cong\mat_n(\mathbb{K})$.
Set $\mm=\mat_n(\mathbb{K})$.
The tangent bundle of the manifold~$\md$ can be identified with the trivial bundle
$\md\times\mm\to\md$.

For $\md=\GL_n(\mathbb{K})$, the YB map~\er{fxyx} is an analytic diffeomorphism of 
the manifold $\md\times\md$.
The differential~$\dft\fmp$ of this diffeomorphism~$\fmp$ can be identified with the following map
\begin{gather}
\lb{dfxyx}
\begin{gathered}
\dft\fmp\cl 
(\md\times\mm)\times(\md\times\mm)
\to
(\md\times\mm)\times(\md\times\mm),\\
\dft\fmp\big((x,M_1),(y,M_2)\big)=\left(\Big(x,M_1\Big),
\Big(xyx^{-1},\frac{\partial}{\partial\varepsilon}\Big|_{\varepsilon=0}
\big((x+\varepsilon M_1)(y+\varepsilon M_2)(x+\varepsilon M_1)^{-1}\big)\Big)\right),
\end{gathered}\\
\notag
x,y\in\md=\GL_n(\mathbb{K}),\qqquad M_1,M_2\in\mm=\mat_n(\mathbb{K}).
\end{gather}
Since $\fmp$ is a YB map, its differential $\dft\fmp$ is a YB map as well.

Let $\Omega\subset\md$ be an abelian subgroup of~$\md$.
Denote by $\ybm\cl(\Omega\times\mm)\times(\Omega\times\mm)\to(\Omega\times\mm)\times(\Omega\times\mm)$ 
the restriction of the map $\dft\fmp$ to the subset 
$(\Omega\times\mm)\times(\Omega\times\mm)\subset(\md\times\mm)\times(\md\times\mm)$.
As $\dft\fmp$ is a YB map, $\ybm$ is a YB map as well.

Let $a,b\in\Omega$. Since $ab=ba$, computing~\er{dfxyx} for $x=a$ and $y=b$, we obtain
\begin{gather}
\lb{ybmgl}
\begin{gathered}
\ybm\colon
\big(\Omega\times\mat_n(\mathbb{K})\big)\times\big(\Omega\times\mat_n(\mathbb{K})\big)\to
\big(\Omega\times\mat_n(\mathbb{K})\big)\times\big(\Omega\times\mat_n(\mathbb{K})\big),\\
\ybm\big((a,M_1),(b,M_2)\big)=\big((a,M_1),
(b,aM_2a^{-1}-bM_1a^{-1}+M_1ba^{-1})\big).
\end{gathered}
\end{gather}
Similarly to Remark~\ref{rnonl},
the YB map~\er{ybmgl} can be interpreted as the following linear parametric YB map
\begin{gather}
\lb{pyblg}
\begin{gathered}
Y_{a,b}\colon\mat_n(\mathbb{K})\times\mat_n(\mathbb{K})\to
\mat_n(\mathbb{K})\times\mat_n(\mathbb{K}),\\
Y_{a,b}(M_1,M_2)=(M_1,aM_2a^{-1}-bM_1a^{-1}+M_1ba^{-1}),
\end{gathered}
\end{gather}
with parameters $a,b\in\Omega$.

Let $l\in\mathbb{K}$, $l\neq 0$.
Applying Corollary~\ref{clc} to the map~\eqref{pyblg}, we obtain the linear parametric YB map
\begin{gather}
\lb{lpyblg}
\begin{gathered}
Y^l_{a,b}\colon\mat_n(\mathbb{K})\times\mat_n(\mathbb{K})\to
\mat_n(\mathbb{K})\times\mat_n(\mathbb{K}),\\
Y^l_{a,b}(M_1,M_2)=(lM_1,l^{-1}aM_2a^{-1}-bM_1a^{-1}+M_1ba^{-1}),
\end{gathered}\\
\notag
a,b\in\Omega,\qquad\text{$\Omega$ is an abelian subgroup of $\GL_n(\mathbb{K})$}.
\end{gather}
In the above construction of~\er{lpyblg} we have assumed
that $\mathbb{K}$ is either $\mathbb{R}$ or $\mathbb{C}$, 
in order to use tangent spaces and differentials.
Now one can check that~\er{lpyblg} is a parametric YB map for any field $\mathbb{K}$.

\section{Some generalisations}
\lb{sgener}

In this section we introduce certain generalisations
of some notions and constructions considered in the paper.

Let $E$, $\Omega$ be topological spaces and $\psi\cl E\to\Omega$ be a fibre bundle.
Consider a YB map $\ybm\cl E\times E\to E\times E$ satisfying $\pi\ybm=\pi$, where
\begin{gather}
\lb{pipsi}
\pi=\psi\times\psi\cl E\times E\to\Omega\times\Omega.
\end{gather}
Such YB maps generalise parametric YB maps discussed in Remark~\ref{rnonl},
which correspond to the case of the trivial bundle $\Omega\times V\to V$, 
where $\Omega$ and $V$ are topological spaces.

Suppose that $\psi\cl E\to\Omega$ is a vector bundle 
and consider the vector bundle $\psi\times\psi\cl E\times E\to\Omega\times\Omega$.
Suppose further that a YB map $\ybm\cl E\times E\to E\times E$ satisfies $\pi\ybm=\pi$
for $\pi$ given by~\er{pipsi}, and $\ybm$ is linear along the fibres of $\psi\times\psi$.
Such YB maps generalise linear parametric YB maps, 
which correspond to the case of the trivial vector bundle $\Omega\times V\to \Omega$, 
where $V$ is a vector space.

Let $M$ be a (smooth or complex-analytic) manifold 
and $Y\cl M\times M\to M\times M$ be a (smooth or complex-analytic) YB map.
Consider the tangent bundle $\tau\cl TM\to M$ of~$M$ and the differential 
$$
\dft Y\cl TM\times TM\to TM\times TM
$$
of the map $Y$. Since $Y$ is a YB map, its differential $\dft Y$ is a YB map as well.
(This follows from general properties of the differentials of smooth and analytic maps.)

For any subset $S\subset M$ satisfying $Y(S\times S)\subset S\times S$, one has 
$$
\dft Y\big(\tau^{-1}(S)\times\tau^{-1}(S)\big)\subset\tau^{-1}(S)\times\tau^{-1}(S), 
$$
and the restriction of the YB map $\dft Y$ to the subset 
$
\tau^{-1}(S)\times\tau^{-1}(S)\subset TM\times TM
$
is a YB map as well.
This generalises the construction from Section~\ref{sYBmvb}, 
which corresponds to the case when $M=\md=\GL_n(\mathbb{K})$, 
where $\mathbb{K}$ is $\mathbb{R}$ or $\mathbb{C}$, 
the map $Y=\fmp$ is given by~\er{fxyx}, and $S=\Omega$ is an abelian subgroup 
of the group $M=\GL_n(\mathbb{K})$.

We plan to study these generalisations in future works.

\section{Conclusions}
\label{sconc}

In this paper we have presented a number of general results 
on linear parametric YB maps,
including clarification of the structure of the nonlinear algebraic relations that define them
and several transformations which allow one to obtain new such maps from known ones.
Also, methods to construct such maps have been described.
In particular, we have demonstrated how to obtain linear parametric YB maps 
(with nonlinear dependence on parameters) 
from nonlinear Darboux transformations of some Lax operators, 
using linear approximations of matrix refactorisation problems corresponding to Darboux matrices.

As illustrative examples, the following new linear parametric YB maps 
with nonlinear dependence on parameters have been presented.
\begin{itemize}
\item
For each nonzero constant $l\in\mathbb{C}$,  
we have the linear parametric 
YB maps~\er{ylab},~\er{lybmss} with parameters~$a,b$.
For $l\neq 1$ these maps are new.
For $l=1$ the map~\er{ylab} appeared in~\cite{dim-mull} 
and \er{lybmss} appeared in~\cite{Sokor-Sasha}.
\item
The map \er{lpyblg} is new for 
each nonzero constant $l\in\mathbb{K}$, 
where $\mathbb{K}$ is any field 
(e.g., $\mathbb{K}=\mathbb{R}$ or $\mathbb{K}=\mathbb{C}$).
Note that \er{lpyblg} actually represents an infinite collection of 
linear parametric YB maps corresponding to 
abelian subgroups $\Omega$ of the matrix groups $\GL_n(\mathbb{K})$, $\,n\in\zsp$.
\end{itemize}

As discussed in Remark~\ref{rinvol},
in terms of dynamics, noninvolutive maps are more interesting than involutive ones,
but known examples of YB maps very often turn out to be involutive.
The maps~\er{ylab},~\er{lybmss},~\er{lpyblg} are noninvolutive.

In Section~\ref{sgener}, using fibre bundles and vector bundles, 
we have introduced certain generalisations
of some notions and constructions considered in the paper.
We plan to study these generalisations in future works.

Also, motivated by the results of this paper,
we suggest the following directions for future research:
\begin{itemize}
\item Develop similar approaches 
for entwining Yang--Baxter maps which are set-theoretical solutions 
to the parametric, entwining Yang--Baxter equation which reads
\begin{gather}\label{entYB}
S^{12}_{a,b}\circ R^{13}_{a,c} \circ T^{23}_{b,c}=T^{23}_{b,c}\circ R^{13}_{a,c} \circ S^{12}_{a,b}.
\end{gather}
Here $S$, $R$, $T$ are maps from $V\times V$ to $V\times V$ depending on two parameters from $\Omega$ 
for some sets $V$, $\Omega$.
The maps $S^{12}_{a,b}$, $R^{13}_{a,c}$, $T^{23}_{b,c}$ from $V\times V\times V$ 
to $V\times V\times V$ with parameters $a,b,c\in\Omega$ are constructed from $S$, $T$, $R$ in the standard way 
(see, e.g.,~\cite{KoulPap11,Pavlos2019,KP2019}).
If $S=R=T\equiv Y$, then equation \eqref{entYB} becomes the parametric YB 
equation~\eqref{YB_eq}. 
\item Extend the methods of this paper to the case 
of the functional (Zamolodchikov's) Tetrahedron equation, 
which can be regarded as a higher-dimensional generalisation of the (parametric) YB equation. 
A number of interesting results on relations of the functional Tetrahedron equation to integrable systems
are known (see, e.g.,~\cite{dim-mull,Kashaev,Kassotakis-Tetrahedron,KRt20,Talalaev}
and references therein); 
however, it has not yet attracted the same attention as the YB equation. 
\end{itemize}

\section*{Acknowledgements} 
This work is supported by the Russian Science Foundation (grant No. 20-71-10110).

We would like to thank A.V.~Mikhailov and D.V.~Talalaev for useful discussions.

\end{document}